\documentclass[12pt]{amsart}
\pdfoutput=1
\usepackage{epsf,psfrag}
\usepackage{enumerate}
\usepackage{epsfig}
\usepackage{graphicx}
\usepackage{graphics}
\usepackage{latexsym}
\usepackage{epsfig}
\usepackage{amsfonts}
\usepackage{amsthm}
\usepackage{amsmath}
\usepackage{amssymb}
\usepackage[all]{xy}
\usepackage{amscd}
\usepackage{mathrsfs}
\usepackage{amsopn}
\usepackage{amscd}
\usepackage{mathrsfs} 
\usepackage{amsopn} 
\usepackage[colorlinks=true]{hyperref}

\newtheorem{theorem}{Theorem}
\newtheorem{prop}[theorem]{Proposition}

\newtheorem{lemma}[theorem]{Lemma}

\newtheorem{cor}[theorem]{Corollary}

\newcommand{\R}{{\mathbb R}}
\newcommand{\Z}{{\mathbb Z}}

\newcommand{\N}{{\mathbb N}}
\newcommand{\T}{{\mathbb T}}
\newcommand{\norm}[1]{\| #1\|}
\newcommand{\trace}{\mathop{\textup{trace}}}
\newcommand{\js}{\op{JointSpec}}

\newcommand{\phy}{\varphi}

\newcommand{\op}[1]{\!\!\mathop{\rm ~#1}\nolimits}

\newcommand{\om}{\omega}

\newcommand{\De}{\Delta}

\renewcommand{\geq}{\geqslant}
\renewcommand{\leq}{\leqslant}
\newcommand{\pscal}[2]{\langle #1,#2\rangle}
\newcommand{\cv}[1]{{\rm Convex\,\, Hull}\,\left(#1\right)}
\newcommand{\abs}[1]{\left|#1\right|}
\newcommand{\OP}{{\rm Op}_\hbar}

\newcommand{\RM}{\mathbb{R}}
\newcommand{\h}{\hbar}

\newcommand{\Id}{{{\rm Id}}}
\newcommand{\f}{{\vec{f}}}
\newcommand{\cA}{{\mathcal{A}}}

\newcommand{\cH}{{\mathcal{H}}}

\newcommand{\cL}{{\mathcal{L}}}

\newcommand{\cO}{{\mathcal{O}}}

\newcommand{\cQ}{{\mathcal{Q}}}

\newenvironment{remark}{\refstepcounter{theorem}\par\medskip\noindent{\bf
Remark~\thetheorem.}}{\unskip\nobreak\hfill\hbox{ $\oslash$}\par\bigskip}

\newenvironment{definition}{\refstepcounter{theorem}\par\medskip\noindent{\bf
Definition~\thetheorem.}}{\unskip\nobreak\hfill\hbox{}\par\medskip}

\begin{document}

\title[Semiclassical quantization and spectral limits]{{\bf Semiclassical
quantization and spectral limits of $\hbar$\--pseudodifferential and
Berezin\--Toeplitz operators}}

\author{{\'A}lvaro Pelayo\,\,\,\,\,\,\, Leonid
  Polterovich\,\,\,\,\,\,\, San V\~{u} Ng\d{o}c}

\date{}

\begin{abstract}
  We introduce a minimalistic notion of semiclassical quantization and
  use it to prove that the convex hull of the semiclassical spectrum of a   quantum
  system given by a collection of  commuting operators converges to the convex hull of 
  the spectrum of the associated classical system. This 
  gives a quick alternative solution to the isospectrality problem for quantum toric 
  systems.  If the operators are uniformly
  bounded, the convergence is uniform. Analogous results hold for non\--commuting operators.
\end{abstract}

\maketitle

\section{Introduction} \label{sec:intro}

In the past ten years there has been a flurry of activity concerning
the interaction between symplectic geometry and spectral theory. The
potential of this interaction was already pointed out by Colin de
Verdi{\`e}re \cite{CdV2, CdV} in a quite general setting: pseudodifferetial operators
on cotangent bundles.  This paper
deals with the question: can one recover a classical system (described
by symplectic geometry) from the spectrum of its quantization
(described by spectral theory)? This is a classical question in inverse spectral
theory going back to pioneer works of Colin de Verdi{\`e}re and 
Guillemin\--Sternberg, in the 1970s and 1980s.

Many contributions followed their works, eg
Iantchenko--Sj{\"o}strand--Zworski~\cite{IaSjZw2002}.  
A few global
spectral results have also been obtained recently, for instance by
Hezari and Zelditch for symmetric domains in
$\R^n$~\cite{hezari-zelditch} or by V\~{u} Ng\d{o}c~\cite{san-inverse}
for one degree of freedom pseudodifferential operators. Recently, the
first and third author settled, jointly with L. Charles, this problem
for toric integrable systems \cite{ChPeVN2011}. The paper
\cite{ChPeVN2011} gives a full description of the semiclassical
spectral theory of toric integrable systems. As a corollary of this theory, that
paper solves the isospectral theorem for toric systems. We refer to
Section \ref{sec:toric} for more details, and to \cite{ChPeVN2011} for
a more complete list of references on inverse spectral results in
symplectic geometry. For an interesting
semiclassical recent work see the article of Guillemin, Paul, and 
Uribe \cite{GuPaUr2007} on trace invariants.

The goal of the present paper is to introduce a 
minimalistic notion of semiclassical quantization
and use it to provide a simple constructive
argument proving that from the joint spectrum of a quantum system one
can recover the convex hull of the classical spectrum. This applies
not only to toric integrable systems, but much more generally.
Moreover, our results are general enough to apply to
pseudodifferential and Berezin-Toeplitz quantization.

\vspace{1mm}
\paragraph{\emph{Structure of the paper}.} In section \ref{results}
we quickly state our main results in the most important particular
cases: pseudodifferential and Berezin-Toeplitz quantization: Theorem~\ref{theo:toeplitz} 
and Theorem~\ref{theo:pseudodifferential}.
In Section \ref{sec-prel-selfad} we recall some preliminaries on
self-adjoint operators.  In Section \ref{sec:quantizable} we introduce
a general quantization setting and state results that encompass
Theorem~\ref{theo:toeplitz} and Theorem~\ref{theo:pseudodifferential}.  In
Section \ref{sec:limits} we prove a key lemma concerning spectral
limits of semiclassical operators.  Section \ref{sec:convexity}
reviews properties of support functions and convex sets which are
needed for the proofs.  In Section \ref{sec:proof} we combine the
previous sections to complete the proofs. In Sections
\ref{sec:toeplitz} and \ref{sec:pseudodifferential} we briefly review
the Berezin\--Toeplitz quantization and $\hbar$-pseudo-differential
calculus and show that they are covered by the quantization procedure
of Section \ref{sec:quantizable}.  

We present a number of examples
illustrating our results: Section \ref{sec:toric} explains how to use
Theorem \ref{theo:toeplitz} to conclude, quickly, the isospectrality
theorem for toric systems which appeared in a previous work
\cite{ChPeVN2011} of the first and third author and L.~Charles. In
Section \ref{sec-examples-1} we apply our results to a physically
interesting system, coupled angular momenta on $S^2 \times S^2$
described by D.A. Sadovski{\'\i} and B.I. Zhilinski{\'\i} in
\cite{SaZh1999}.  Section \ref{subsec-rot} deals with
$\hbar$-pseudodifferential quantization of a particle on the plane in
a rotationally symmetric potential.

In the last section we extend our results to nonnecessarily
commuting operators.

\section{Main results} \label{results}

For simplicity,
in this section we will only state our results in the two (otherwise
rather general) cases of pseudodifferential and Berezin-Toeplitz quantization.
In the statements below, we fix a quantization scheme for a symplectic
manifold $M$ and focus on a collection $\mathcal{F} = (T_1,\dots,T_d)$
of mutually commuting self-adjoint semi-classical operators. These
operators depend on the Planck constant $\hbar \in I$, where $I$ is a
subset of $(0,1]$ that accumulates at $0$, and act on a Hilbert space
$\mathcal{H}_{\hbar},\; \hbar\in I$.

Let $T_1,\ldots,T_d$ are pair-wise commuting selfadjoint (not
necessarily bounded) operators on a Hilbert space $\mathcal{H}$ with a
common dense domain $\mathcal{D} \subset \mathcal{H}$ such that
$T_j(\mathcal{D}) \subset \mathcal{D}$ for all $j$.  
The \emph{joint spectrum} of $(T_1,\ldots,T_d)$ is by definition
the support of the joint spectral measure. It is denoted by
$\op{JointSpec}(T_1,\, \ldots, \, T_d)$.

For instance, if $T_j$'s are endomorphisms of a finite dimensional
vector space, then the joint spectrum of $T_1,\ldots,T_d$ is the set
of $(\lambda_1,\dots,\lambda_d)\in \mathbb{C}^d$ such that there
exists a non-zero vector $v$ for which
\[
T_j v = \lambda_j v, \quad \forall j=1,\dots,n.
\]
If $T_1,\, \ldots,\, T_d$ are pairwise commuting semiclassical
operators, then of course the joint spectrum of $T_1,\ldots,T_d$
depends on $\hbar$.

\begin{figure}[h]
  \includegraphics[width=0.8\textwidth]{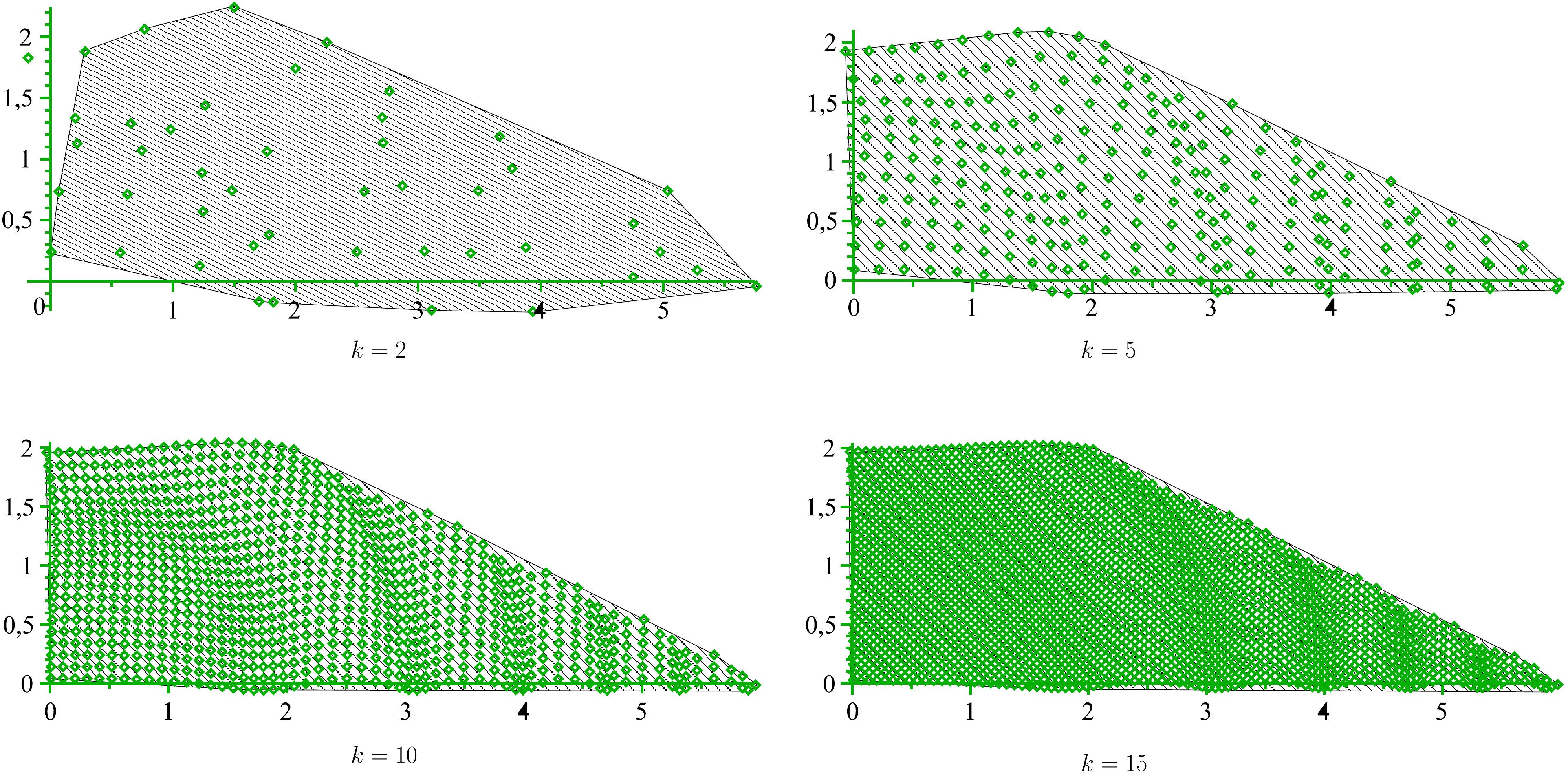}
  \caption{Convergence of convex hulls of semiclassical 
  spectra ($k=1/\hbar$) of a quantum  toric system \cite{ChPeVN2011}.}
  \label{fig:convex_hull}
\end{figure}

Following the physicists, we shall call \emph{classical spectrum} of
$(T_1,\dots, T_d)$ the image $F(M)\subset \R^d$, where
$F=(f^{(1)},\dots,f^{(d)})$ is the map of principal symbols of
$T_1,\dots,T_d$.

\begin{theorem} \label{theo:toeplitz} Let $M$ be a prequantizable
  closed symplectic manifold.  Let $d\geq 1$ and let
  $\mathcal{F}:=(T_1,\dots, T_d)$ be a family of pairwise commuting
  selfadjoint Berezin-Toeplitz operators on $M$.  Let
  $\mathcal{S}\subset\R^d$ be the classical spectrum of $\mathcal{F}$.
  Then
  \begin{equation}\label{eq-Hausdorff-spectrum}
    \lim_{\hbar \rightarrow 0} {\rm Convex\,\, Hull}\,
    \Big(\textup{JointSpec}(\mathcal{F}) \Big) \,=\, \cv{\mathcal{S}}
  \end{equation}
  where the limit convergence is in the Hausdorff metric.
\end{theorem}

In the $\hbar$-pseudodifferential setting (see for
instance~\cite{dimassi-sjostrand}), the analysis is more complicated
due to the possible unboundedness of the operators. Here we work with
the standard H\"{o}rmander class of symbols $a(x,\xi,\hbar)$ on
$\RM^{2n}$ or ${\rm T}^*X$ with closed $X$ imposing the following
restriction on the dependence of $a$ on $\hbar$:
$$a(x,\xi,\hbar)=a_0(x,\xi) + \hbar a_{1,\hbar}(x,\xi)\;,$$
where all $a_{1,\hbar}(x,\xi)$ are uniformly (in $\hbar$) bounded and
supported in the same compact set.  The principal symbol $a_0$ can be
unbounded. We shall say that $a$ {\it mildly depends on $\hbar$}.
Fortunately, in a number of meaningful examples one deals with
$\hbar$-independent symbols, see e.g. Section \ref{subsec-rot} below.

\begin{theorem} \label{theo:pseudodifferential} Let $X$ be either
  $\mathbb{R}^n$, $n \geq 1$, or a closed manifold. Let $d\geq 1$ and
  let $\mathcal{F}:=(T_1,\dots T_d)$ be a family of pairwise commuting
  selfadjoint $\hbar$-pseudodifferential operators on $X$ whose
  symbols mildly depend on $\hbar$. Then the following hold.
  \begin{itemize}
  \item[{(i)}] From the family
    \[
    \Big\{ {\rm Convex\,\, Hull}\,\Big(\textup{JointSpec}
    (\mathcal{F})\Big) \Big\}_{\hbar \in J}
    \]
    one can recover the convex hull of the classical spectrum of
    $\mathcal{F}$.
  \item[{(ii)}] If, in addition, the principal symbols of $T_j$ are
    bounded for every $1 \leq j \leq d$, the joint and the classical
    spectra are related by equation \eqref{eq-Hausdorff-spectrum}.
  \end{itemize}
\end{theorem}

\medskip
\noindent
The assumption on mild dependence of the symbol on $\hbar$ cannot be
completely dropped, see Remark \ref{remark-hbar-dependence} below.  An
immediate consequence of the theorems is the following simple
statement.

\begin{cor} \label{cor:im0} Let $T_1,\dots,T_d$ be either:
  \begin{itemize}
  \item[{(i)}] commuting self\--adjoint Berezin\--Toeplitz operators on a
    prequantizable closed symplectic manifold, or
  \item[{(ii)}] $\hbar$\--pseudodifferential self\--adjoint operators on
    $\mathbb{R}^n$ or a closed manifold whose symbols mildly depend on
    $\hbar$.
  \end{itemize}
  If the classical spectrum of $(T_1,\dots, T_d)$ is convex, then the
  joint spectrum recovers it.
\end{cor}

\section{Preliminaries on self-adjoint operators}\label{sec-prel-selfad}
First let us recall some elementary results from operator theory. For
any (not necessarily bounded) selfadjoint operator $A$ on a Hilbert
space with a dense domain and with spectrum $\sigma(A)$, we have
\begin{equation}
  \sup \sigma(A) = \sup_{u\neq 0} \frac{\pscal{Au}{u}}{\pscal{u}{u}}\;.
  \label{equ:rayleigh}
\end{equation}
This follows for instance from~\cite[Proposition
5.12]{hislop-sigal}. We give the proof here for the reader's
convenience.  If $\lambda\in\sigma(A)$, then by the Weyl criterium,
there exists a sequence $(u_n)$ with $\norm{u_n}=1$ such that
\[
\lim_{n\to\infty} \norm{(A-\lambda {\rm Id})u_n}=0.
\]
Therefore, $\lim_{n\to\infty}\pscal{Au_n}{u_n} = \lambda$, which
implies that
\[
\sup_{\norm{u}=1} \pscal{Au}{u} \geq \sup \sigma(A).
\]
Conversely, if $\sigma(A)$ lies in $(-\infty,c]$, then by the Spectral
Theorem $A \leq c \cdot \Id$ which yields
\[
\sup_{\norm{u}=1} \pscal{Au}{u} \leq \sup \sigma(A).
\]
This proves \eqref{equ:rayleigh}.

Let us mention also that it implies
\begin{equation}
  \sup \{|s|: s \in \sigma(A)\} = \sup_{u\neq 0} \frac{|\pscal{Au}{u}|}{\pscal{u}{u}} = \norm{A}\leq \infty ;.
  \label{equ:rayleigh-1}
\end{equation}

\section{Semiclassical quantization} \label{sec:quantizable}

We shall prove Theorem~\ref{theo:toeplitz}
and Theorem~\ref{theo:pseudodifferential} in a more general context of
semiclassical quantization.

Let $M$ be a connected manifold (either closed or open).  Let
$\mathcal{A}_0$ be a subalgebra of $\textup{C}^\infty(M;\R)$
containing the constants and all compactly supported functions. We fix
a subset $I \subset (0,1]$ that accumulates at $0$.  If $\mathcal{H}$
is a complex Hilbert space, we denote by $\mathcal{L}(\mathcal{H})$
the set of linear (possibly unbounded) selfadjoint operators on
$\mathcal{H}$ with a dense domain. By a slight abuse of notation, we
write $\norm{T}$ for the \emph{operator norm} of an operator, and $\norm{f}$ for
the \emph{uniform norm} of a function on $M$.

\begin{definition} \label{BT} A \emph{semiclassical quantization} of
  $(M,\mathcal{A}_0)$ consists of a family of complex Hilbert spaces
  $\mathcal{H}_{\hbar},\; \hbar\in I$, and a family of $\R$\--linear
  maps $\OP \colon \mathcal{A}_0 \to \mathcal{L}(\mathcal{H}_{\hbar})$
  satisfying the following properties, where $f$ and $g$ are in
  $\mathcal{A}_0$:
  \begin{enumerate}[(Q1)]
  \item \label{item:one} $\norm{\OP (1) - {\rm Id}} =
    \mathcal{O}(\hbar)$ {\bf (normalization)};
  \item \label{item:garding} for all $f \geq 0$ there exists a
    constant $C_f$ such that $\OP (f) \geq -C_f \hbar $ {\bf
      (quasi-positivity)};
  \item \label{item:norm} let $f\in\mathcal{A}_0$ such that $f\neq 0$
    and has compact support, then
  $$\liminf_{\hbar\to 0}
  \norm{\OP(f)}>0$$ {\bf (non-degeneracy)};
\item \label{item:symbolic} if $g$ has compact support, then for all
  $f$, $\OP(f)
  \circ \OP(g)$ is bounded, and we have \[ \norm{\OP(f) \circ \OP(g) -
    \OP(fg)} = \mathcal{O}(\hbar),
  \] {\bf (product formula)}.
\end{enumerate}
A \emph{quantizable} manifold is a manifold for which there exists a
semiclassical quantization.
\end{definition}

\medskip
\noindent We shall often use the following consequence of these
axioms: for a bounded function $f$, the operator $\OP(f)$ is
bounded. Indeed, if $c_1 \leq f \leq c_2$ for some $c_1,c_2 \in \RM$,
(Q\ref{item:one}) and (Q\ref{item:garding}) yield
\begin{equation}\label{eq-gard-cor}
  c_1\cdot\Id - \mathcal{O}(\hbar) \leq \OP(f) \leq c_2\cdot \Id +  \mathcal{O}(\hbar).
\end{equation}
Since our operators are selfadjoint, this implies by formula
\eqref{equ:rayleigh-1} above
\begin{equation} \label{eq-norm-bounds-vsp} \norm{\OP(f)} \leq \norm{f} +
  \cO(\hbar)\;.
\end{equation}

Next, consider the algebra $\mathcal{A}_I$ whose elements are
collections $\vec{f} =(f_\hbar)_{\hbar \in I}$, $f_\hbar
\in \mathcal{A}_0$ with the following property: for each $\f$ there
exists $f_0 \in \mathcal{A}_0$ so that
\begin{equation}\label{eq-vector}
  f_\hbar = f_0+ \hbar f_{1,\hbar}\;,
\end{equation}
where the sequence $f_{1,\hbar}$ is uniformly bounded in $\hbar$ and
supported in the same compact set $K=K(\f) \subset M$.  The function
$f_0$ is called the {\it principal part} of $\f$. If $f_0$ is
compactly supported as well, we say that $\f$ is \emph{compactly
  supported}.

\begin{definition}
  We define a map $$\text{Op}: \mathcal{A}_I \to \prod_{\hbar \in I}
  \mathcal{L}(\mathcal{H}_\hbar),\;\; \f=(f_\hbar) \mapsto
  (\OP(f_\hbar))\;.$$ {\it A semiclassical operator} is an element
  in the image of $\text{Op}$. Given $\f \in \cA_I$, the function
  $f_0 \in \cA_0$ defined by \eqref{eq-vector} is called {\it the
    principal symbol} of $\text{Op}(\f)$.
\end{definition}

By \eqref{eq-norm-bounds-vsp}
\begin{equation}\label{eq-cO}
  \OP(f_\hbar) = \OP(f_0)+ \cO(\hbar)\;.
\end{equation}
This together with the product formula (Q\ref{item:symbolic}) readily
yields that for every $\vec{g}$ with compact support and
every $\f$,
\begin{equation}
  \label{eq-symbolic-prime}
  \norm{\OP(f_\hbar) \circ \OP(g_\hbar) -
    \OP(f_\hbar g_\hbar)} = \mathcal{O}(\hbar)\;.
\end{equation}

Now we are ready to show that {\it the principal symbol of a
  semiclassical operator is unique}.  Indeed, if $\text{Op}(\f)=0$,
then for any compactly supported function $\chi$, we get by
\eqref{eq-symbolic-prime}
\[
\OP(f_\hbar\chi) = \OP(f_\hbar)\OP(\chi) + \mathcal{O}(\hbar) =
\mathcal{O}(\hbar),
\]
and then by \eqref{eq-cO}, $\OP(f_0\chi) =
\mathcal{O}(\hbar)$. By~(Q\ref{item:norm}) we conclude that
$f_0\chi=0$. Since $\chi$ is arbitrary, $f_0=0$.

\medskip

Important examples of semiclassical quantization in the sense of
Definition~\ref{BT} are provided by symplectic geometry.  They include
Berezin-Toeplitz quantization on closed prequantizable symplectic manifolds and
(certain versions of) $\hbar$-pseudodifferential calculus on cotangent
bundles. In the latter case the algebra $\cA_0$ is the usual
H\"{o}rmander class of symbols, while its deformation $\cA_I$ is more
special and is defined as above. The details will be explained in
Sections~\ref{sec:toeplitz} and~\ref{sec:pseudodifferential}.

\begin{remark}
  The definition of semiclassical quantization above is
  ``minimalistic": we list only those axioms which will enable us to
  reconstruct the classical spectrum from the quantum one, which is
  the main objective of the paper. In particular, we do not require a
  stronger form of the correspondence principle which links together
  commutators of semiclassical operators with Poisson brackets of
  their symbols, and even do not assume that the manifold $M$ is
  symplectic.
\end{remark}

For the following results recall the notion of quantizable manifold we use 
(Definition \ref{BT}).

\begin{theorem} \label{theo:general} Let $M$ be a quantizable
  manifold.  Let $d\geq 1$ and let $(T_1,\dots T_d)$ be pairwise
  commuting semiclassical operators on $M$.  Let $J$ be a subset of
  $I$ that accumulates at $0$.  Then from the family
  \[
  \Big\{ {\rm Convex\,\, Hull}\,\Big(\textup{JointSpec}
  (T_1,\dots,T_d)\Big) \Big\}_{\hbar \in J}
  \]
  one can recover the convex hull of the classical spectrum of
  $(T_1,\dots,T_d)$.
\end{theorem}

We can strengthen Theorem~\ref{theo:general} in case the principal
symbols are bounded to obtain a uniform convergence in the Hausdorff
distance.

Recall that the \emph{Hausdorff distance} between two subsets $A$ and
$B$ of $\R^d$ is
 $$
 {\rm d}_H(A,\,B):= \inf\{\epsilon > 0\,\, | \,\ A \subseteq
 B_\epsilon \ \mbox{and}\ B \subseteq A_\epsilon\},
$$
where for any subset $X$ of $\R^d$, the set $X_{\epsilon}$ is
$$X_\epsilon := \bigcup_{x \in X} \{m \in \R^d\, \, | \,\, \|x - m \|
\leq \epsilon\},$$ (see eg. \cite{BuBuIv2001}).

\begin{theorem} \label{theo:bounded} Let $M$ be a quantizable
  manifold.  Let $d\geq 1$ and let $(T_1,\dots T_d)$ be pairwise
  commuting semiclassical operators on $M$. Let $J$ be a subset of $I$
  that accumulates at $0$.  Assume that the principal symbols of
  $T_j$, $j=1,\dots,d$, are bounded. Let
  $\mathcal{S}\subset\R^d$ be the classical spectrum of $(T_1,\dots,
  T_d)$.  Then
$$
\lim_{\hbar \rightarrow 0} {\rm Convex\,\, Hull}\,
\Big(\textup{JointSpec}(T_1,\dots, T_d) \Big) \,=\, \cv{\mathcal{S}}
$$
where the limit convergence is in the Hausdorff metric.
\end{theorem}

\medskip
\noindent Theorem \ref{theo:general} and Theorem \ref{theo:bounded} together
with results of Sections \ref{subsec-bt-prelim} and
\ref{subsec-pdo-prelim} below readily yield Theorems
\ref{theo:toeplitz} and \ref{theo:pseudodifferential} of the
introduction.

\medskip

\begin{remark} 
  Let $T$ be a uniformly (in $\hbar$) bounded selfadjoint operator
  with principal symbol $f$. Applying Theorem~\ref{theo:bounded} with
  $d=1$, we get that $$ \lim_{\hbar \rightarrow 0} [\lambda_{\min}(T),
  \lambda_{\max}(T)] \,=\, {\mathcal{S}} = f(M),
$$
where $\lambda_{\min}(T)$ and $\lambda_{\max}(T)$ are the minimum and
maximum of the spectrum of $T$, respectively. Since by
\eqref{equ:rayleigh-1} the operator norm of a bounded selfadjoint
operator equals $\max (\abs{\lambda_{\min}},\abs{\lambda_{\max}})$, we
get that the axioms of semiclassical quantization listed in Definition
\ref{BT} above imply the following ``automatic" refinement of
non-degeneracy axiom (Q\ref{item:norm}):
\begin{equation} \label{eq-norm-improved}
\lim_{\hbar \to 0} \norm{\OP(f)} = \norm{f}
\end{equation}
for every bounded function $f \in \mathcal{A}_0$.
\end{remark}

\medskip

\begin{cor} \label{cor:im} If the classical spectrum is convex, then
  the joint spectrum recovers it.
\end{cor}

Corollary \ref{cor:im} is an immediate consequence of Theorem
\ref{theo:general}. Note that we don't need to know the precise
structure of the joint spectrum of the quantum system in order to
recover the classical spectrum -- it suffices to know the convex hull
of this joint spectrum, as a subset of $\mathbb{R}^d$.

\section{Spectral limits} \label{sec:limits}

Fix a semiclassical quantization the sense of Definition~\ref{BT} on a
manifold $M$.

\begin{lemma} \label{lem:max} Take any $\f= (f_\hbar)
  \in\mathcal{A}_I$ with principal part $f_0$, and let
  $(\OP(f_\hbar))$ be the corresponding semiclassical operator.  Let
  $\lambda_{\sup}(\hbar)$ denote the supremum of the spectrum of
  $\OP(f_\hbar)$.  Then
  \begin{eqnarray} \label{eq:0} \lim_{\hbar \to 0}
    \lambda_{\sup}(\hbar)= \sup_M f_0.
  \end{eqnarray}
\end{lemma}

\begin{proof}
  For clarity we divide the proof into several steps.

  \paragraph{\emph{Step 1}.} By \eqref{equ:rayleigh} and \eqref{eq-cO}
$$ \lambda_{\sup}(\hbar)=  \sup_{\norm{u}=1}{\pscal{\OP(f_\hbar) u}{u}}=
\sup_{\norm{u}=1}{\pscal{\OP(f_0) u}{u}}=
\lambda_{\sup}(\OP(f_0))\;.$$ Therefore it suffices to prove the lemma
assuming that $f_\hbar =f_0$ for all $\hbar$.  From now on, this will
be the standing assumption until the end of the proof.

Further, fix $\epsilon >0$ sufficiently small. We claim that
\begin{equation}
  \lambda_{\sup}(\hbar) \leq \sup_M f_0 + \epsilon
  \label{equ:supless}
\end{equation}
for all $\hbar$ sufficiently small.  Indeed, if $f_0$ is unbounded
from above there is nothing to prove. If $f_0$ is bounded from above,
this follows from \eqref{eq-gard-cor}.

\paragraph{\emph{Step 2}.} Put
\begin{equation}
  \label{equ:F}
  F_\epsilon :=
  \begin{cases}
    \sup_M f_0 - \epsilon &\text{ if } f_0 \text{ is bounded from above;}\\
    1/\epsilon &\text{ otherwise}.
  \end{cases}
\end{equation}

Let $K$ be a compact set with non-empty interior and such that
$$
f_0|_K \geq F_\epsilon.
$$
Let $\chi \ge 0$ be a smooth function which is identically $0$ outside
of $K$, identically $1$ in a compact $\tilde{K} \subset \mathring{K}$
with non\--empty interior. Then the function
$$
( f_0 - F_\epsilon ) \,\, \chi
$$
is $0$ outside of $K$, and, inside of $K$, it is greater than or equal
to $0$.  Then, by Axiom (Q\ref{item:garding}), there
exists a positive constant $C_{\epsilon}$ such that
\begin{eqnarray} \label{eq:-1} \langle \OP((f_0 - F_\epsilon)
  \chi^2)u,\, u \rangle \geq -C_\epsilon\hbar.
\end{eqnarray}
As above, in what follows $C_\epsilon$ denotes a positive constant
which does not depend on $\hbar$ and whose value may vary from step to
step.

\paragraph{\emph{Step 3}.} We claim that there exists $u \in
\mathcal{H}_{\hbar}$ such that $\| u \|=1$ and
\begin{eqnarray} \label{eq:u} u = \OP(\chi) u +\mathcal{O}(\hbar).
\end{eqnarray}
Indeed, let $\tilde{\chi}$ be supported on $\tilde{K}$ and non
identically $0$. By Axiom (Q\ref{item:norm}) there exists a constant
$c>0$ such that $\| \OP(\tilde{\chi}) \| \ge c >0$.  Therefore there
exists some $v \in \mathcal{H}_{\hbar}$ with $\|v\|=1$ and such that
\begin{equation}\label{eq-c}
  \|
  \OP(\tilde{\chi}) v \| > \frac{c}{2}\;.
\end{equation}
Now let
$$
u=\frac{\OP(\tilde{\chi}) v}{\| \OP(\tilde{\chi}) v \|}.
$$
By the product formula (Axiom~(Q\ref{item:symbolic})) we know that
\[
\OP(\chi \tilde \chi)=\OP(\chi) \circ
\OP(\tilde{\chi})+\mathcal{O}(\hbar),
\]
and therefore
\begin{eqnarray} \label{eq:3} \OP(\chi)u=
  \frac{\OP(\chi)\OP(\tilde{\chi})v}{\|\OP(\tilde{\chi}) v \|} =
  \frac{\OP(\chi \tilde{\chi})v} {\| \OP(\tilde{\chi}) v \|} +
  \mathcal{O}(\hbar).
\end{eqnarray}
In the second equality above we use \eqref{eq-c}. Since $\tilde{\chi} \chi
=\tilde{\chi}$, it follows from equation (\ref{eq:3}) that
$$
\OP(\chi)u=\frac{\OP(\tilde{\chi})v} {\|\OP(\tilde{\chi}) v \|} +
\mathcal{O}(\hbar)= u +\mathcal{O}(\hbar),
$$
which proves the claim (\ref{eq:u}).

\paragraph{\emph{Step 4}.} Put $w:= \OP(\chi)u$, where $u$ is from Step 3.
By selfadjointness and Axiom~(Q\ref{item:symbolic})
\begin{eqnarray}
  \langle \OP(f_0 -F_\epsilon) w,\, w \rangle &=&
  \langle \OP(\chi)\circ\OP(f_0 -F_\epsilon)\circ\OP(\chi)u,u\rangle \nonumber \\
  &=&
  \langle \OP((f_0 -F_\epsilon)\chi^2)u,u\rangle + \mathcal{O}(\hbar) \geq -C_\epsilon \hbar\;, \nonumber
\end{eqnarray}
where the last inequality follows from \eqref{eq:-1}.  Using
(Q\ref{item:one}) and the fact that $\|w\|=1 +\cO(\hbar)$ by
\eqref{eq:u}, we conclude that
$$
\lambda_{\sup}(\hbar) \geq F_\epsilon -C_{\epsilon} \hbar.
$$
Now, if $\hbar$ is small enough then $C_{\epsilon} \hbar < \epsilon$,
and hence, in view of~\eqref{equ:F},
$$
\lambda_{\sup}(\hbar) \geq \sup_M f_0 -\epsilon\;,
$$
if $\sup_M f_0 < +\infty$ and
$$
\lambda_{\sup}(\hbar) \geq \epsilon^{-1} -\epsilon\;,
$$
if $\sup_M f_0 = +\infty$.  Since $\epsilon>0$ is arbitrary, this
together with~\eqref{equ:supless} implies (in both cases) that
$$
\lim_{\hbar \to 0} \lambda_{\sup}(\hbar)= \sup_M f_0,
$$
as required
\end{proof}

\section{Detecting convexity} \label{sec:convexity}

\begin{figure}[h]
  \centering
  \includegraphics[width=0.6\textwidth]{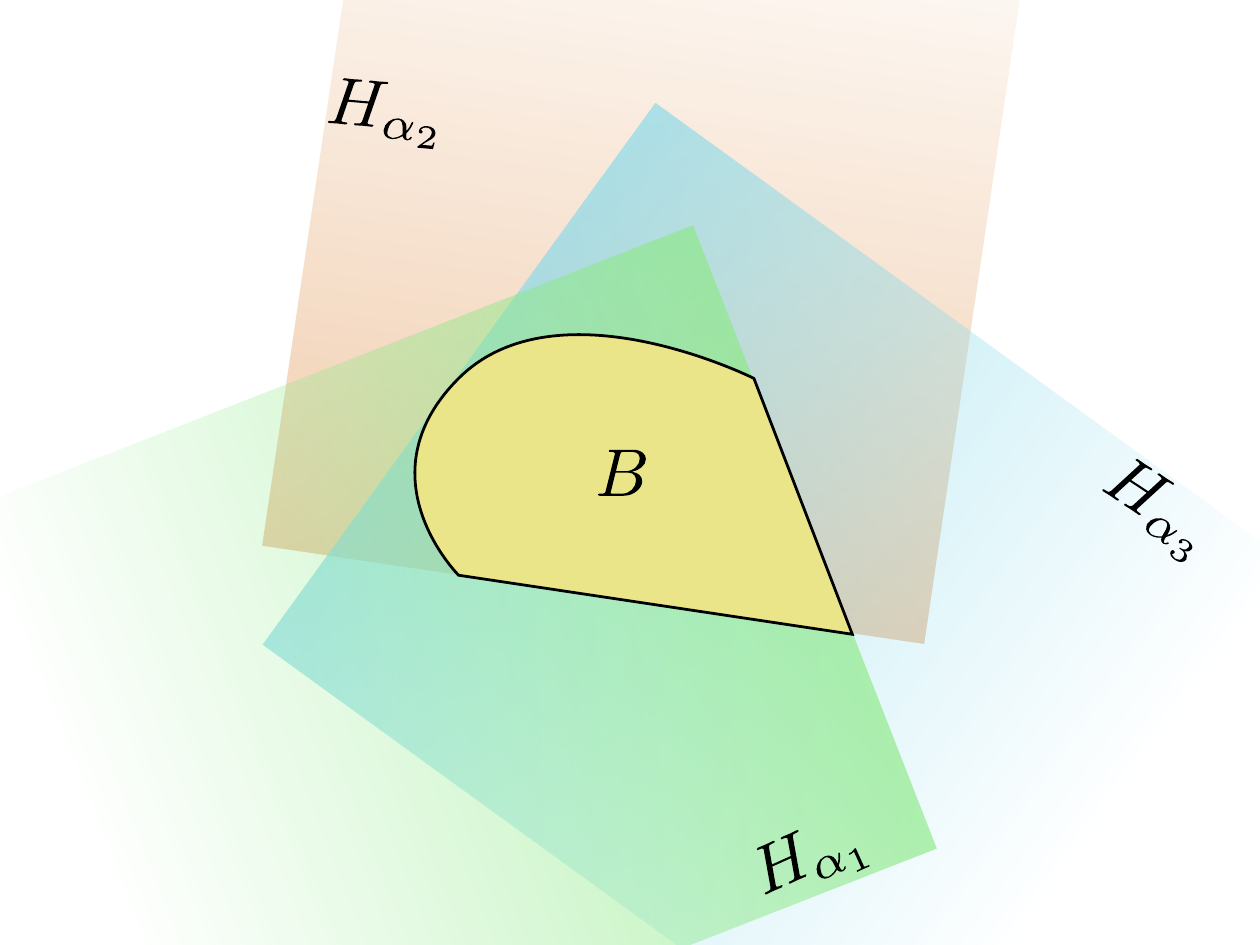}
  \caption{Lemma \ref{lem:duality} in the case of $2$ dimensions.}
  \label{fig:2dlem}
\end{figure}

Let $B \subseteq \mathbb{R}^d$ be a closed set. Let $S^{d-1}$ be the
unit sphere in $\mathbb{R}^d$.  The map $ \Phi_B \colon S^{d-1} \to
\mathbb{R} $ given by
\[
\Phi_B(\alpha):=\sup_{x \in B} \langle x,\, \alpha \rangle \in \R \cup
\{+\infty\} \] is called {\it the support function} of $B$ (here we
deviate a little bit from the standard definition where $\Phi_B$ is
defined on the whole $\R^d$).  The following facts are well known (see
e.g. \cite[Section 7.2]{Bauschke-Combettes}):

\begin{lemma} \label{lem:duality}\begin{itemize}
  \item[{(i)}] $ {\rm Convex\,\, Hull}\,(B) = \bigcap_{\alpha \in
      S^{d-1}(E)} \{x \in \mathbb{R}^d \, | \, \langle x, \alpha
    \rangle \leq \Phi_B(\alpha)\}$.
  \item[{(ii)}] $\Phi_B=\Phi_{{\rm Convex\,\, Hull}\,(B)}$.
  \end{itemize}

\end{lemma}

\begin{prop} \label{prop:phi} $\;$
  \begin{itemize}
  \item[{(i)}] Let $A$ and $B$ be closed sets. Then we have the
    following equivalence~:

$$  \left(\Phi_A \leq \Phi_B\right) \Longleftrightarrow
\left(\cv{A} \subset \cv{B}\right).$$

\item[{(ii)}] Let $A$ be a convex closed set. Let $\epsilon>0$. Then
$$
\Phi_A +\epsilon \,\, = \,\, \Phi_{A+\overline{{\rm B}(0,\epsilon)}}
$$
\item[{(iii)}] Let $B$ be a compact set, and let $C\geq 0$ such that
  for all $x \in B, \quad \norm{x}\leq C$. Then $\Phi_B$ is
  $C$-Lipschitz.
\end{itemize}
\end{prop}

\begin{proof}
  
  \medskip
  \noindent {\bf (i)} If $\Phi_A\leq \Phi_B$, then
  Lemma~\ref{lem:duality}(i) gives $\cv{A} \subset
  \cv{B}$. Conversely, if $\cv{A} \subset \cv{B}$, then by definition
  of the maps $\Phi$, we have $\Phi_{\cv{A}}\leq \Phi_{\cv{B}}$. We
  conclude by Lemma \ref{lem:duality}(ii).
    
  \medskip
  \noindent {\bf (ii)} If $x\in A$ and $b\in \overline{{\rm
      B}(0,\epsilon)}$, we have
  \[
  \pscal{\alpha}{x+b} \leq \pscal{\alpha}{x} + \epsilon,
  \]
  and hence $\Phi_{A+\overline{{\rm B}(0,\epsilon)}} \leq \Phi_A
  +\epsilon$. Conversely, note that $\epsilon =\langle \alpha,
  \epsilon\alpha \rangle$. Therefore
  \begin{eqnarray}
    \langle \alpha, x \rangle +\epsilon &=& \langle \alpha, x \rangle + \langle \alpha, \epsilon \alpha \rangle \nonumber  \\
    &=&\langle \alpha, \underbrace{x+\epsilon \alpha}_{\in A +
      \overline{{\rm B}(0,\epsilon)}}  \rangle \nonumber \\
    &\leq& \sup_{x \in A+\overline{{\rm B}(0,\epsilon)}}  \pscal{\alpha}{x}
  \end{eqnarray}
  which concludes the proof.
    
  \medskip
  \noindent {\bf (iii)} If $\alpha,\alpha'\in S^{d-1}$, we have
  \[
  \abs{ \pscal{x}{\alpha} -\pscal{x}{\alpha'} } \leq \|x\|
  \|\alpha-\alpha'\|,
  \]
  which easily implies
$$
|\Phi_B(\alpha)-\Phi_B(\alpha') | \leq C \| \alpha-\alpha'\|\;.
$$
\end{proof}

\section{Proof of Theorems \ref{theo:general} and
  \ref{theo:bounded}} \label{sec:proof}

\begin{lemma} \label{lem:btcommuting} Let $T_1,\dots,T_d$ be pairwise
  commuting selfadjoint (possibly unbounded) operators on a Hilbert
  space. For any $\alpha\in S^{d-1}$, let $\sigma(\alpha)$ be the
  spectrum of $T^{(\alpha)}:=\sum_{j=1}^d \alpha_j\, T_j$.  Then for
  any fixed $\alpha$,
  \begin{eqnarray}
    && \sup  \sigma(\alpha)
    = \sup \{ \langle x,\,  \alpha \rangle \, | \, x \in {\rm JointSpec}(T_1,\ldots, T_d)\} \label{equ:psi-JS}\\
    &&= \sup \{ \langle x,\,  \alpha \rangle \, | \, x \in  {\rm Convex\,\, Hull}\,({\rm JointSpec}(T_1,\ldots, T_d))\}.  \nonumber
  \end{eqnarray}
\end{lemma}

\begin{figure}[h]
  \centering
  \includegraphics[width=0.6\textwidth]{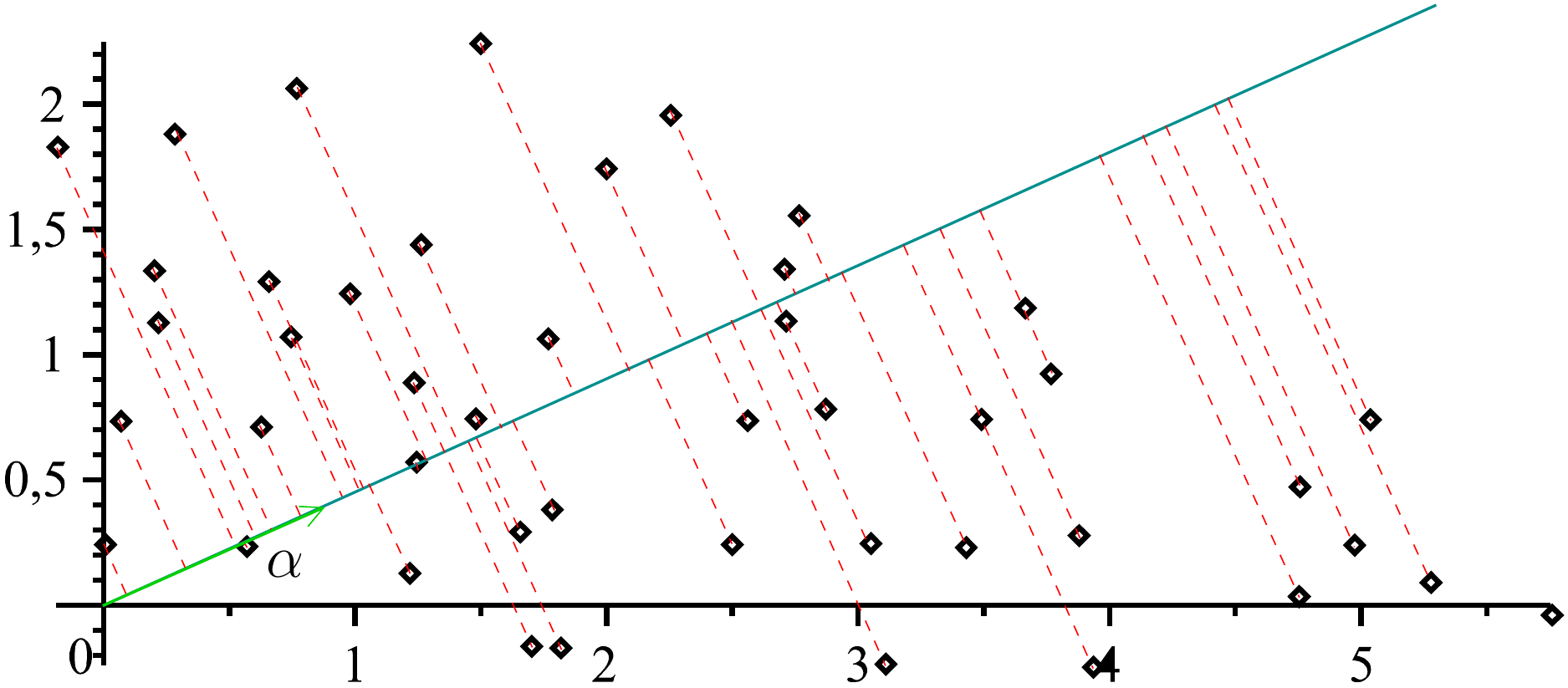}
  \caption{The projection of the joint spectrum onto the line directed
    by the vector $\alpha$ gives the spectrum of $\sum_{j=1}^d
    \alpha_j T_j$}
  \label{fig:2dlem}
\end{figure}

\begin{proof} Let $\mu_j$ be the spectral measure of $T_j$, and let
  $\mu=\mu_1 \otimes \dots \otimes \mu_d$ be the joint spectral
  measure on $\R^d$.  For a given $\alpha \in S^{d-1}$ define a linear
  functional $\phi: \RM^d \to \RM$ by $\phi(x) = \langle x,\alpha
  \rangle$. Observe that
$$T^{(\alpha)}= \int_{\RM^d} \phi \;{\rm d}\mu = \int_{\RM} t {\rm d}(\phi_*\mu)\;,$$
where $\phi_*\mu$ is the push-forward of $\mu$ to $\RM$. By definition
of the support of a spectral measure, $\phi(\text{supp}(\mu))$ is a
dense subset of $ \text{supp}(\phi_*\mu)$.  Thus $$\{ \langle x,\,
\alpha \rangle \, | \, x \in {\rm JointSpec}(T_1,\ldots, T_d)\} $$ is
a dense subset in $\sigma(\alpha)$. This proves the first equality in
\eqref{equ:psi-JS}.  The second equality follows from Lemma
\ref{lem:duality}(ii).
\end{proof}

Now we are ready to give a proof of Theorem~\ref{theo:general}
and Theorem~\ref{theo:bounded}.

\medskip
\paragraph{\emph{Proof of Theorem \ref{theo:general}}.}

We denote by $\Sigma$ the joint spectrum of $(T_1,\dots,T_d)$. Let
$\alpha\in S^{d-1}\subset \R^d$ and let $f_{\alpha}:=\langle \alpha,
\, F \rangle$. By Lemma~\ref{lem:max}, we have
\[
\lim_{\hbar \to 0} \lambda_{\max}({\rm T}_{f_\alpha}) = \sup f_\alpha,
\]
which, in view of~\eqref{equ:psi-JS}, reads
\begin{equation}
  \lim_{\hbar \to 0} \Phi_{\Sigma}(\alpha) = \Phi_{F(M)}(\alpha).
  \label{equ:pointwise}
\end{equation}
Therefore, the map $\Phi_{F(M)}$ can be recovered from the joint
spectrum $\Sigma$; by Lemma~\ref{lem:duality} we can then recover
$\cv{F(M)}$, which proves Theorem~\ref{theo:general}.

\medskip

\paragraph{\emph{Proof of Theorem~\ref{theo:bounded}.}}

If the principal symbols of $T_1,\dots,T_d$ are bounded then the joint spectrum $\Sigma$ is 
 bounded, and it follows from Proposition~\ref{prop:phi}(iii) that the
family of maps $$(\Phi_{\Sigma}- \Phi_{F(M)})_{\hbar \in (0,h_0]}$$ is
uniformly equicontinuous for $h_0>0$ small enough.

This uniform equicontinuity and the compactness of the sphere
$S^{d-1}$ imply that the pointwise limit~\eqref{equ:pointwise} is in
fact uniform~:
\[
\forall \epsilon >0 , \exists \hbar_0>0 , \forall \alpha \in S^{d-1} ,
\forall \hbar<\hbar_0,\,\,\,\, \abs{\Phi_{\Sigma}(\alpha) -
  \Phi_{F(M)}(\alpha)} \leq \epsilon.
\]

Now Proposition~\ref{prop:phi}(i),(ii) gives the inclusions
\begin{equation*} \label{A} \cv{{\rm JointSpec}(T_1,\ldots, T_d)}
  \subset \cv{F(M)} + \overline{{\rm B}(0,\epsilon)}
\end{equation*}
and
\begin{equation*} \label{B} \cv{F(M)} \subset \cv{{\rm
      JointSpec}(T_1,\ldots, T_d)} + \overline{{\rm B}(0, \epsilon}).
\end{equation*}

In other words, if $\hbar\leq \hbar_0$, the Hausdorff distance between
the joint spectrum of $(T_1,\ldots, T_d)$ and $F(M)$ is less than
$\epsilon$, which proves Theorem~\ref{theo:bounded}.

\section{Berezin\--Toeplitz quantization} \label{sec:toeplitz}

 \subsection{Preliminaries}\label{subsec-bt-prelim}  In this section we are concerned with Berezin\--Toeplitz operators
 (simply called in the sequel Toeplitz operators) and quantization of
 classical systems given by such operators, see Kostant \cite{Ko1970},
 Souriau \cite{So1970}, and Berezin \cite{Be1975}, as well as the book
 \cite{BoGu1981} by Boutet de Monvel and Guillemin for the
 corresponding microlocal analysis. Many well known results for
 pseudodifferential operators are now known for Toeplitz operators,
 see \cite{BMS,guillemin_star-product,BU,MM,Ch2003b}.

 Let $(M,\omega)$ be a closed symplectic manifold whose symplectic
 form represents an integral de Rham cohomology class of $M$ times
 $2\pi$. In what follows such symplectic manifolds will be called {\it
   prequantizable}. When a prequantizable symplectic manifold
 $(M,\omega)$ is K\"{a}hler with respect to a complex structure $J$,
 the Berezin-Toeplitz quantization can be given by the following
 geometric construction.  Choose a holomorphic Hermitian line bundle
 $\cL$ over $M$ so that the curvature of its (unique) Hermitian
 connection compatible with the holomorphic structure equals
 $-{\rm i}\omega$ (the existence of such a $\cL$ is a well known fact from
 complex algebraic geometry).  For a positive integer $m=1/\hbar$,
 \begin{equation}
   \mathcal{H}_{\hbar}:=\mathrm{H}^0(M,\mathcal{L}^m)
   \label{equ:hilbert}
 \end{equation}
 is the space of holomorphic sections of $\mathcal{L}^m$.

 Since $M$ is compact, $\mathcal{H}_{\hbar}$ is a finite dimensional
 subspace of the Hilbert space ${\rm L}^2(M, \mathcal{L}^m)$. Here the
 scalar product is defined by integrating the Hermitian pointwise
 scalar product of sections against the Liouville measure of $M$.
 Denote by $\Pi_{\hbar}$ the orthogonal projector of ${\rm L}^2(M,
 \mathcal{L}^m)$ onto ${\mathcal{H}}_{\hbar}$.

 Put $\cA_0= \textup{C}^\infty(M)$ and define the quantization map
 $\OP$ by $$\OP(f) = \Pi_{\hbar}S_f,$$ where $S_f$ is the operator of
 multiplication by $f$.  Here the Planck constant $\hbar$ runs over
 the set $I=\{{\textstyle\frac{1}{m}}\;|\; m\in\N\}$.  The fact that
 $\OP$ is a semiclassical quantization in the sense of
 Definition~\ref{BT} is proved in \cite{BMS}.

 \medskip

 In what follows we shall use that the Berezin-Toeplitz quantization
 behaves in a functorial way with respect to direct products of closed
 K\"{a}hler manifolds: Given prequantum bundles $\mathcal{L}_1$ and
 $\mathcal{L}_2$ over $M_1$ and $M_2$ respectively, form a prequantum
 bundle $\mathcal{L}= p_1^*\mathcal{L}_1 \otimes p_2^*\mathcal{L}_2$
 over $M:= M_1 \times M_2$, where $p_i:M \to M_i$ is the natural
 projection. By a version of K\"{u}nneth formula \cite{Kunneth}
 ${\rm H}^0(M,\mathcal{L}) = {\rm H}^0(M_1,\mathcal{L}_1)\otimes
 {\rm H}^0(M_2,\mathcal{L}_2)$. For a pair of functions $f^{(i)} \in
 \textup{C}^\infty(M_i)$, $i=1,2$ define a new function $f \in
 \textup{C}^\infty(M)$ by $f(x_1,x_2):= f^{(1)}(x_1)f^{(2)}(x_2)$ One
 can readily check that
 \begin{equation}\label{BT-product}
   \OP(f) = \OP(f_1) \otimes \OP(f_2)\;.
 \end{equation}

 \medskip

 If $M$ is a prequantizable closed symplectic (but not necessarily
 K\"{a}hler) manifold, there exist several constructions of
 semiclassical quantizations of $M$ satisfying Definition \ref{BT},
 see \cite{BoGu1981, guillemin_star-product, BU, ShZe2002, MM}.

 Now we are ready to present some specific applications of our main
 results in the context of the Berezin-Toeplitz quantization.

 \subsection{The case of Hamiltonian torus actions} \label{sec:toric}

 Assume that $M$ is a prequantizable K\"ahler manifold endowed with a
 Hamiltonian $\T^d$-action which preserves the complex structure. Then
 the Kostant-Souriau formula yields commuting Toeplitz operators whose
 principal symbols are the components of the $\T^d$-momentum map. A
 proof when $d=n$ was outlined in Charles\--Pelayo\--V\~{u} Ng\d{o}c
 \cite[Theorem~1.4]{ChPeVN2011} as a byproduct of the main result of
 the paper (a normal form theorem for quantum toric systems); the
 $d\leq n$ case was stated in Charles~\cite[section 3]{Ch2006a}
 without further details. Therefore such a quantum $\T^d$ action, $d
 \leq n$, satisfies the hypothesis of our main theorem.

 \vskip 1em

 Not all symplectic manifolds have a complex structure or a prequantum
 bundle. However a symplectic toric manifold, that is when $2d=2n$ is
 the dimension of $M$, always admits a compatible complex structure,
 which is not unique. Furthermore a symplectic toric manifold $M$ with
 momentum map $\mu:M\to\mathbb{R}^n$ is prequantizable if and only if
 there exists $c \in \mathbb{R}^n$ such that the vertices of the
 polytope $\mu(M)+c$ belong to $2 \pi \mathbb{Z}^n$. If it is the
 case, the prequantum bundle is unique up to isomorphism.
 Figure~\ref{fig:convex_hull} shows the joint spectrum of a
 4-dimensional toric manifold (a Hirzebruch surface).

 By the Atiyah and Guillemin-Sternberg theorem \cite{At1982, GuSt1982}
 , for any Hamiltonian torus action on a connected closed manifold,
 the image of the momentum map is a rational convex
 polytope~\cite{At1982,GuSt1982}.  So the map
 $\mu=(\mu_1,\dots,\mu_n):M\to\mathbb{R}^n$ satisfies the assumption
 of Corollary~\ref{cor:im}. Even more, for a symplectic toric
 manifold, the momentum polytope $\Delta \subset \R^n $ has the
 additional property that for each vertex $v$ of $\Delta$, the
 primitive normal vectors to the facets meeting at $v$ form a basis of
 the integral lattice $\Z^n$. We call such a polytope a {\em Delzant
   polytope}.

\begin{definition}
  Two symplectic toric manifolds $(M, \om ,\mu)$ and $( M', \om' ,
  \mu')$ are \emph{isomorphic} if there exists a symplectomorphism
  $\varphi: M \rightarrow M'$ such that $\varphi^* \mu' = \mu.$
\end{definition}

By the Delzant classification theorem \cite{d}, a symplectic toric
manifold is determined up to isomorphism by its momentum
polytope. Furthemore, for any Delzant polytope $\Delta$, Delzant
constructed in \cite{d} a symplectic toric manifold $(M_{\Delta} ,
\om_{\Delta} , \mu_{\Delta} )$ with momentum polytope $\Delta$.

\begin{cor}[Isospectrality for toric systems,
  \cite{ChPeVN2011}] \label{theo:inverse-spectral} Let $T_1,\dots,
  T_n$ be commuting self\--adjoint Toeplitz operators on a symplectic
  toric manifold $(M, \, \omega, \, \mu : M \rightarrow \R^n)$ whose
  principal symbols are the components of $\mu$. Then
$$\De : = \lim_{\hbar \rightarrow 0}  \textup{JointSpec}(T_1,\dots, T_n)$$
is the Delzant polytope $\mu(M)$. Moreover, $(M, \om, \mu)$ is
isomorphic with $(M_{\Delta}, \om_{\Delta}, \mu_{\Delta})$.
\end{cor}

The approach of the present paper bypasses the precise description of
the semiclassical spectral theory, so it is less informative than the
one of \cite{ChPeVN2011}. However, it has the advantage of concluding
isospectrality with an easier proof, which moreover applies in a much
more general setting.

\subsection{Coupled angular momenta} \label{sec-examples-1}

Here we present an example of a non-toric integrable system modeling a
pair of coupled angular momenta as a prequantum bundle of $(S^2, \frac{1}{2}\sigma)$.  It has been described first by
D.A. Sadovski{\'\i} and B.I. Zhilinski{\'\i} in \cite{SaZh1999} (see
also \cite[Example 6.2]{VN2007} and \cite{PeVN2009, PeVN2011} for
further discussion).

In order to present this system we need some preliminaries. Consider
the unit sphere $S^2 \subset \R^3$ equipped with the standard area
form $\sigma$ of total area $4\pi$. Let $(x,y,z)$ be the Euclidean
coordinates on $\R^3$ considered as functions on $S^2$.  The Poisson
brackets of these functions satisfy the relation
\begin{equation}\label{eq-comrel-class}
  \{x,y\} = z
\end{equation}
and its cyclic permutations.

Identify $S^2$ with the complex projective line $\mathbb{C} P^1 =
\mathbb{C} \cup \{\infty\}$ by the map
$$W \in  \mathbb{C} \cup \{\infty\} \mapsto \Big{(} \frac{2\text{Re}\;W}{1+|W|^2}, \frac{2\text{Im}\;W}{1+|W|^2}, \frac{1-|W|^2}{1+|W|^2}\Big{)} \in S^2\;.$$ Let $\mathcal{L}$ be the holomorphic line bundle over $\mathbb{C} P^1$
dual to the tautological one. We fix the scheme $T$ of the
Berezin-Toeplitz quantization associated to $\mathcal{L}$ considered
as a prequantum bundle of $(S^2, \frac{1}{2}\sigma)$. By a straightforward but cumbersome
calculation with Berezin's coherent states \cite{Be1975} one can
verify the quantum commutation relation
\begin{equation}\label{eq-comrel-quant} [T_m(x),T_m(y)] =
  -\frac{2i}{m+2} T_m(z)
\end{equation}
and its cyclic permutations.

For a positive number $a$, the sphere $(S^2,a\sigma)$ serves as the
phase space of classical angular momentum whose components are given
by $(ax, ay, az)$. The number $a$ plays the role of the amplitude of
the angular momentum. In view of \eqref{eq-comrel-class} we have the
relation
\begin{equation}\label{eq-comrel-class-1}
  \{x,y\}_a = a^{-1}z
\end{equation}
and its cyclic permutations, where $\{.,.\}_a$ stands for the Poisson
bracket associated to $a\sigma$.  If $a$ is a positive half-integer,
that is $a \in \mathbb{N}/2$, the sphere $(S^2,a\sigma)$ is
quantizable with the prequantum bundle $\mathcal{L}^{2a}$. The
corresponding Berezin-Toeplitz quantization $T^{(a)}$ is given by
\begin{equation}\label{eq-BT-a}
  T_m^{(a)}(f) = T_{2am}(f)\;.
\end{equation}

Fix now $a_1,a_2 >0$. The phase space of the system of coupled angular
momenta is the manifold $M= S^2 \times S^2$ equipped with the
symplectic form $\omega = a_1\sigma_1 \oplus a_2\sigma_2$, while the
coupling Hamiltonian is independent on $a_1,a_2$ and is given by $$H=
x_1x_2+y_1y_2 +z_1z_2.$$ Here and below we equip all the data
corresponding to the first and the second factor of $M$ by lower
indices $1$ and $2$ (e.g. $x_2$ is the $x$-coordinate on the second
factor, etc.) We write $\{.,.\}_M$ for the Poisson bracket on $M$.
The coupling Hamiltonian $H$ admits a first integral $F= a_1z_1 +
a_2z_2$. Indeed, by using \eqref{eq-comrel-class-1} one readily checks
that $\{H,F\}_M =0$.

In order to quantize this system, introduce the prequantum bundle $L =
p_1^*\mathcal{L}_1^{2a_1} \otimes p_2^*\mathcal{L}_2^{2a_2}$ over $M$,
where $\mathcal{L}_j$ is a copy of $\mathcal{L}$ over the $j$-th
factor of $M$, and $p_j$ is the projection of $M$ to the $j$-th
factor, $j=1,2$.  Denote by $\widehat{T}_m$ the corresponding
Berezin-Toeplitz quantization.

For $j=1,2$ put $\gamma_{j,m}:= 1+ a_j^{-1}m^{-1}$, and set
\begin{eqnarray}
X_j&:=& \gamma_{j,m}\widehat{T}_m(x_j),  \nonumber \\
Y_j&:=&\gamma_{j,m}\widehat{T}_m(y_j), \nonumber \\
Z_j&:=& \gamma_{j,m}\widehat{T}_m(z_j)\;. \nonumber 
\end{eqnarray}
By \eqref{BT-product} and \eqref{eq-BT-a}, the operators
$$\widehat{H}_m:= X_1 \otimes X_2 + Y_1 \otimes Y_2 + Z_1 \otimes Z_2$$
and
$$\widehat{F}_m:= a_1 Z_1 \otimes \Id +  \Id \otimes a_2Z_2$$
are Toeplitz operators with the principal symbols $H$ and $F$
respectively.  The commutation relations \eqref{eq-comrel-quant}
readily yield that $[\widehat{H},\widehat{F}]=0$.  Let us emphasize
that for a given integrable system $F_1,\dots,F_n$, the existence of
pair-wise commuting semiclassical operators with principal symbols
$F_1,\dots,F_n$ is not at all automatic, see \cite{ChPeVN2011} for a
discussion and references.

\medskip

Next, let us describe the classical spectrum of the system, that is
the image of the momentum map
$$\Phi: S^2 \times S^2 \to \RM^2,\; (v,w) \mapsto (F(v,w), H(v,w))\;.$$
Without loss of generality assume that $a_1 =1$ and $a_2=a \geq 1$
(otherwise make a rescaling, maybe after the permutation of the
variables). It is not hard to see that the image of the map
$$\Psi: S^2 \times S^2 \to \RM^3,\; (v,w) \mapsto v+aw$$ is the spherical shell
$$\{u \in \R^3\;:\; a-1 \leq |u| \leq a+1\}\;.$$
Observe that
$$H= \frac{1}{2a} \cdot (|\Psi|^2 -a^2-1)\;,$$
and $F$ is simply the $z$-coordinate of $\Psi$. Therefore on each
sphere of radius $$r: = |\Psi|= \sqrt{1+a^2+2aH} \in [a-1,a+1]\;,$$ in
$\RM^3$ centered at the origin, the value of $F$ runs from $-r$ to
$r$. Furthermore, $H$ is an increasing function of $r$ which takes
values $\pm 1$ at $r= a \pm 1$. We conclude that the image of
$\Phi=(F,H)$ is the domain $\Delta \subset \RM ^2$ given by the
inequalities
$$F^2 \leq 1+a^2+2aH,\;\; -1 \leq H \leq 1\;.$$
This domain is clearly convex.

We conclude by Theorem \ref{theo:toeplitz} that {\it the convex hull
  of the joint spectrum of the Toeplitz operators $\widehat{F}_m$ and
  $\widehat{H}_m $ converges to $\Delta$ in the Hausdorff sense as $m
  \to \infty$.}

\section{$\hbar$-Pseudodifferential quantization}
\label{sec:pseudodifferential}

\subsection{Preliminaries}\label{subsec-pdo-prelim}  

If $M=\R^{2n}$ or $M$ is the cotangent bundle of a closed manifold,
then a well-known semiclassical quantization of $M$ is given by
$\hbar$-pseudodifferential operators, which is a semiclassical variant
of the standard homogeneous pseudodifferential operators (see for
instance~\cite{dimassi-sjostrand}). In this setting, commuting
semiclassical pseudodifferential operators have been considered by
Charbonnel~\cite{charbonnel}; see also~\cite{san-cpam}. In the
remaining of this text, we omit the $\hbar-$ prefix for notational
simplicity.

Symbolic calculus of pseudodifferential operators is known to hold
when the symbols belong to a Hörmander class. For instance one can
take
\[
\mathcal{A}_0 := \{ f \in \textup{C}^\infty(\R^{2n}_{(x,\xi)}) \,\,
:\, \, \exists m\in\R\quad \abs{\partial_{(x,\xi)}^\alpha f}\leq
C_\alpha\langle (x,\xi)\rangle^m \quad \forall \alpha\in\N^{2n}\}.
\]
Here $\langle z \rangle:= (1+|z|^2)^{1/2}$ for $z \in \R^q$.  If
$f\in\mathcal{A}_0$, its Weyl quantization is defined on
$\mathcal{S}(\R^n)$ by
\begin{equation}
  \label{equ:weyl}
  (\OP(f) u)(x) := \frac{1}{(2\pi\hbar)^{n}}\int_{\R^{2n}}
  \textup{e}^{\frac{\textup{i}}{\hbar}((x-y)\cdot\xi)} f({\textstyle\frac{x+y}{2}},\xi)u(y)\textup{d}y \textup{d}\xi.
\end{equation}

Let $X$ be a closed $n$-dimensional manifold equipped with a smooth
density $\mu$.  We cover it by a finite set of coordinate charts
$U_1,\dots,U_N$ each of which is identified with a convex bounded
domain of $\R^n$ equipped with the Lebesgue measure (the existence of
such an atlas readily follows from Moser's argument
\cite{Moser-Dacarogna}). Let $\chi_1^2,\dots,\chi_N^2$ be a partition
of unity subordinated to $U_1,\dots,U_N$, that is $\text{supp}(\chi_j)
\subset U_j$ and $$\sum_j \chi_j^2 =1.$$ Then, for any
$f\in\textup{C}^\infty(\textup{T}^*X)$ such that
\[
\abs{\partial_\xi^\alpha f(x)} \leq C_\alpha \langle \xi \rangle ^m,
\qquad \forall (x,\xi)\in \textup{T}^*X, \;\forall \alpha\in \N^n,
\]
for some $m\in\R$, we define, for $u\in\textup{C}^\infty(X)$,
\begin{equation} \label{eq-weyl-manifold} \OP(f) u:= \sum_{j=1}^N
  \chi_j \cdot \OP^j(f) (\chi_j u)\;,
\end{equation}
where $\OP^j(f)$ is the Weyl quantization calculated in $U_j$.  The
operator $u\mapsto \chi_j \cdot \OP^j(f)(\chi_j u)$ is a
pseudodifferential operator on $X$ with principal symbol
$f(x,\xi)\chi_j^2(x)$ for $(x,\xi)\in \textup{T}^*_x X$. Therefore
$\OP(f)$ is a pseudodifferential operator on $X$ with principal symbol
$\sum f\chi_j^2 = f$. The standard pseudodifferential symbolic
calculus \cite{dimassi-sjostrand,Z} gives the following proposition.
\begin{prop}
  $\hbar$-pseudodifferential quantization on $X$, where $X$ is either
  $\R^n$ or a closed manifold equipped with a density, is a
  semiclassical quantization of $\textup{T}^*X$ in the sense of
  Definition~\ref{BT}, where $\mathcal{A}_0$ is a Hörmander symbol
  class, $I=(0,1]$, the Hilbert space $\mathcal{H}_\hbar$ is
  $\textup{L}^2(X)$ (it is independent of $\hbar$). If $X=\R^n$,
  $\OP(f)$ is given by the Weyl quantization. If $X$ is a closed
  manifold, $\OP(f)$ is constructed via formula
  \eqref{eq-weyl-manifold}.
\end{prop}

\begin{proof}
  For the reader's convenience, let us outline the proof of the
  proposition. Let $S_j: \textup{C}^\infty (U_j) \to
  \textup{C}^\infty_0 (U_j)$ be the operator of multiplication by
  $\chi_j$. Then \eqref{eq-weyl-manifold} reads $\OP(f) =\sum_j S_j
  \OP^j(f) S_j$, and so $\OP(f)$ is self-adjoint since $S_j$ and the
  Weyl quantization are self-adjoint.  Axiom (Q\ref{item:one}) follows
  from the fact that $S_j\OP^j(1)S_j=S_j^2$, and $\sum_j S_j^2 = {\rm
    Id}$. Axiom (Q\ref{item:garding}) is known as the G{\aa}rding
  inequality \cite[Theorem 7.12]{dimassi-sjostrand}.  Axiom
  (Q\ref{item:symbolic}) is a consequence of the product formula for
  the Weyl quantization (see e.g. \cite[Theorem
  7.9]{dimassi-sjostrand}).

  For property (Q\ref{item:norm}) see \cite[Theorem 13.13]{Z}.
  Alternatively, it is not hard to derive it from a standard result on
  spectral asymptotics for $\hbar$-pseudo-differential operators (see
  e.g. \cite[Corollary 9.7]{dimassi-sjostrand}).  For reader's
  convenience, we present a direct short argument.

  Let $P$ be a pseudo-differential operator with a compactly supported
  principal symbol $p \neq 0$.  Take $(x_0,\xi_0)\in {\rm T}^*X$ such
  that $p(x_0,\xi_0)\neq 0$.  Let $\chi$ be a smooth function on $X$
  supported in a small neigbourhood of $x_0$, contained in a
  coordinate chart $U$ for $X$, and such that
  $\norm{\chi}_{L^2(X)}=1$. In these coordinates, we define the WKB
  function on $U$
  \begin{equation}
    u_\hbar(x) := {\rm e}^{\frac{{\rm i}}{\hbar}\pscal{x}{\xi_0}}\chi(x).\label{equ:wkb}
  \end{equation}
  Then we may compute directly
  \[
  (\chi P u_\hbar) (x) = \frac{1}{(2\pi\hbar)^n}\int_{\RM^{2n}}{\rm
    e}^{\frac{{\rm
        i}}{\hbar}\pscal{x-y}{\xi}}a_\hbar({\textstyle\frac{x+y}{2}},\xi)
  {\rm e}^{\frac{{\rm i}}{\h}\pscal{y}{\xi_0}}\chi(y)\chi(x){\rm
    d}y{\rm d}\xi.
  \]
  The phase $$\phy(x,y,\xi)=\pscal{x-y}{\xi} + \pscal{y}{\xi_0}$$ is
  stationary when $\partial_y\phy = \xi_0-\xi=0$ and $\partial_\xi\phy
  = x-y=0 $. The hessian $${\rm d}^2\phy=
  \begin{pmatrix}
    0 & -{\rm Id}\\ -{\rm Id} & 0
  \end{pmatrix}
  $$ is non-degenerate.  Thus, for fixed $x$ in a neighborhood $V$ of
  $x_0$ the stationary phase approximation gives
  \[
  \abs{(\chi P u_\hbar) (x)} = \abs{a_\h(x,\xi_0)\chi(x)} +
  \mathcal{O}(\h) = \abs{p(x,\xi_0)\chi(x)} + \mathcal{O}(\h).
  \]
  Hence
  \begin{equation}
    \label{equ:leading}
    \abs{P u_\hbar (x)} = \abs{p(x,\xi_0)}+ \mathcal{O}(\h).
  \end{equation}
  Furthermore, the remainder in \eqref{equ:leading} is of order $
  \mathcal{O}(\h)$ uniformly in $x \in V$. Therefore, shrinking the
  neighborhood $V$ if necessary, we get that
  \[
  \forall x\in V,\qquad \abs{P u_\hbar (x)}> \abs{p(x_0,\xi_0)}/2
  \] for $\h$ small enough. This gives the desired result, since
  \[
  \norm{P} \geq \norm{P u_\h}_{{\rm L}^2(X)} \geq
  \abs{p(x_0,\xi_0)}\sqrt{\text{Volume}{V}}/2.
  \]
  This completes the proof.
\end{proof}

\begin{remark}\label{remark-hbar-dependence} {\rm  In Theorem \ref{theo:pseudodifferential}(ii) and
Corollary \ref{cor:im0} the assumption on mild dependence of the symbol
 on $\hbar$ cannot be  dropped. Indeed, consider a single $\hbar$-pseudodifferential 
operator $\OP(\hbar x)= \hbar x$ on $\RM$. Here the dependence of  $a(x,\xi,\hbar)=\hbar x$ on $\hbar$
is not mild. The principal symbol vanishes so the classical spectrum equals $\{0\}$, while
for each $\hbar >0$ the spectrum of $\OP(\hbar x)$ equals $\RM$. We conclude that the classical spectrum
cannot be recovered from the quantum one in the classical limit. Let us mention that
 the formulation and the proof of Theorem \ref{theo:general} remain valid if we allow a slightly more general
 class of symbols 
$$a(x,\xi,\hbar)=a_0(x,\xi) + \hbar a_{1,\hbar}(x,\xi)\;,$$
where all $a_{1,\hbar}(x,\xi)$ are uniformly (in $\hbar$) bounded but
not necessarily compactly supported.  The significance of such a
generalization is questionable since such symbols do not form an
algebra. It contains however two subspaces that are algebras: $\cA_I$
as defined in the present paper, and $\cA'_I$ which is defined by
requesting $a_0$ to be bounded and $a_{1,\hbar}$ to be uniformly
bounded.}
\end{remark}

\subsection{Example: particle in a rotationally symmetric
  potential}\label{subsec-rot}

Consider a particle on the plane $\RM^2(x_1,x_2)$ with potential
energy $V(x_1^2 + x_2^2)$, where $V$ is a smooth function on
$[0,+\infty)$ such that $V(0)=0$, $V >0,V'>0$ and $V''\geq 0$ on
  $(0,+\infty)$.

The Hamiltonian of the particle is $$H(x,\xi)=
\frac{\xi_1^2+\xi_2^2}{2} + V(x_1^2 + x_2^2).$$ The corresponding
Hamiltonian system admits a first integral, the angular momentum
$F(x,\xi) = x_1 \xi_2 -x_2\xi_1\;$.  Fix a H\"{o}rmander class $\cA_0$
and assume that it contains $H$. A standard calculation
with the Weyl quantization shows that
\begin{eqnarray}
  \OP(H) &=& -\frac{1}{2}\hbar^2 \cdot \Big{(}\frac{\partial^2}{\partial x_1^2} + \frac{\partial^2}{\partial x_2^2}\Big{)} + V(x_1^2+x_2^2)\;; \nonumber \\
  \OP(F)&=& -{\rm i}\hbar \Big{(}x_1 \cdot \frac{\partial}{\partial x_2}- x_2 \cdot \frac{\partial}{\partial x_1}\Big{)}\;.
  \nonumber
\end{eqnarray}
The operators $\OP(H),\OP(G)$ commute.

\begin{lemma}\label{lem-Leg} 
   Let
  \begin{equation}\label{eq-Leg-g}
    g(z):= \max_{r\geq 0} (rz- rV(r))
  \end{equation}
  be the Legendre transform of $rV(r)$. Then $g(0)=0$, $g>0$ on
  $(0,+\infty)$, and $\sqrt{g}$ is concave on $(0, +\infty)$.
\end{lemma}
\begin{proof}
  It is immediate that $g(0)=0$ and $g>0$ on $(0,+\infty)$. Now assume
  $z>0$. The maximum of $rz-rV(r)$ is attained at a unique point, say
  $r(z)$. By a property of the Legendre transform,
  \begin{equation}\label{eq-Leg-1}
    g'(z) = r(z)\;,
  \end{equation}
  and hence
  \begin{equation}\label{eq-Leg-2}
    g''(z) = r'(z)\;.
  \end{equation}
  Observe also that $r'>0$: indeed, $g$ is strictly convex since $rV$
  is strictly convex on $(0,+\infty)$.  Further,
  \begin{equation}\label{eq-Leg-vsp}
    z= (rV)'(r(z)) = r(z)V'(r(z)) + V(r(z))\;.
  \end{equation}
  Differentiating by $z$ we get that
  \begin{equation}\label{eq-Leg-3}
    1=  2r'(z)V'(r(z))+r(z) r'(z) V''(r(z)) \geq 2r'(z)V'(z)\;,
  \end{equation}
  where the last inequality follows from $V''\geq 0$.  By
  \eqref{eq-Leg-vsp},
  \begin{equation}\label{eq-Leg-vsp-1}
    g(z) = r(z)(z-V(r(z))= r^2(z)V'(r(z))\;.
  \end{equation}

  In order to prove concavity of $\sqrt{g}$ it suffices to check that
  $(\sqrt{g})'' \leq 0$, that is $(g')^2 \geq 2gg''\;$. By
  \eqref{eq-Leg-1},\eqref{eq-Leg-2} and \eqref{eq-Leg-vsp-1} this is
  equivalent to $r^2(z) \geq 2r^2(z)V'(r(z))r'(z)$, but this follows
  from equation \eqref{eq-Leg-3}.
\end{proof}

\begin{prop}\label{prop-rot}
  Assume that $V(0)=0$ and $V >0,V'>0$ and $V''\geq 0$ on
  $(0,+\infty)$.  Then the classical spectrum of $((\OP(H)),(\OP(F)))$
  is convex, and hence by Corollary \ref{cor:im0}(ii) it can be
  recovered from the joint spectrum of $(\OP(H),\OP(F))$ in the
  classical limit.
\end{prop}

\begin{proof}
  Put $s=\frac{\xi_1^2 +\xi_2^2}{2}$ and $r=x_1^2+x_2^2$. Fix an
  energy level $\Sigma_z= \{H=z\}$, that is $s= z- V(r)$. Let $g(z)$
  be the maximal possible value of $rs$ on this level, that is $g$ is
  given by \eqref{eq-Leg-g}. Let $r(z)$ be as in the lemma above.
  Since $$\{|x|=\sqrt{r(z)}, \; |\xi| = \sqrt{2(z-V(r(z)))}\} \subset
  \Sigma_z,$$ the value of $F(x,\xi)= x_1 \xi_2 -x_2\xi_1$ runs between
  $-|x|\cdot |\xi|=-\sqrt{2g(z)}$ and $|x|\cdot |\xi|=\sqrt{2g(z)}$.
  It follows that the classical spectrum is the subset of
  $\mathbb{R}^2$ defined by $ |F| \leq \sqrt{2g(H)}$, which, by Lemma
  \ref{lem-Leg}, a is convex set.
\end{proof}

\section{Non\--commuting operators}

Let $\cH$ be a separable Hilbert space, and let $T_1,\dots,T_d$ be
possibly unbounded selfadjoint operators on $\cH$. 
Denote by
$\cQ$ the set of all \emph{mixed quantum states} for
$T_1,\dots,T_d$~: the elements of $\cQ$ are trace class selfadjoint operators
$Q\in\cL(\cH)$ such that
\begin{enumerate}
\item $Q\geq 0$;
\item $\trace{Q} = 1$;
\item $T_jQ \in \cL^1$, for all $j=1,\dots, d$.
\end{enumerate}
Here $\cL^1$ stands for the Schatten trace class, and 
the image of $Q$ is assumed to lie in the domain of $T_j$, 
$j=1,\ldots,d$. Of course $\cQ$ is
not a vector space, but it is a convex set of compact operators.  The
third condition is void in case the operators $T_j$'s are bounded:
indeed, the Hölder inequality gives $$\norm{T_j Q}_{\cL^1}\leq
\norm{T_j}\norm{Q}_{\cL^1}$$ (see also~\cite[Theorem
VI.19]{reed-simon-1}).

\begin{definition} \label{ex} The map
  \begin{eqnarray} \label{expec:map} E: \cQ \to \R^d, \;\; Q \mapsto
    (\trace(T_1Q),\dots, \trace(T_dQ))\;,
  \end{eqnarray}
  is called the {\it expectation map}. We denote by $
  \Sigma(T_1,\dots,T_d) \subset \mathbb{R}^d $ the closure of the image
  of $E$.
\end{definition}

An interesting subset of $\cQ$ consists of the so-called \emph{pure
  states}, i.e. rank-one orthogonal projectors. We shall call the
closure of the image of pure states by $E$ the \emph{joint numerical
  range} $R$~:
\[
R := \overline{\{(\pscal{T_1 u}{u}, \dots, \pscal{T_d u}{u}); \quad u\in\cH,
  \norm{u}=1\}} \subset  \Sigma(T_1,\dots,T_d).
\]

\begin{lemma}
  \label{lemm:numerical-range}
  $ \Sigma(T_1,\dots,T_d)$ is the convex hull of $R$.
\end{lemma}

\begin{proof}
  Let $Q\in\cQ$, and let $(e_n)_{n\geq 0}$ be a Hilbert basis of
  eigenvectors of $Q$ so that $Qe_n=q_ne_n$.
   For any $j=1,\dots, d$, we can write
  \[
  \trace{T_jQ} = \sum_{n\geq 0} \pscal{T_jQ e_n}{e_n} = \sum_{n\geq 0}
  q_n \pscal{T_j e_n}{e_n},
  \]
  where $q_n\geq 0$ and $\sum_n q_n = 1$. Therefore any element of
  $E(\cQ)$ is of the form $\sum_n q_n \vec z_n$, where
\[
\vec z_n := (\pscal{T_1 e_n}{e_n},\dots \pscal{T_d e_n}{e_n}) \in R.
\]
Being an infinite convex combination of elements of $R$, $\vec z_n$
must lie in the closed convex hull of $R$, which is equal to the
convex hull of $R$ since $R$ is closed, by definition. Thus
\[
 \Sigma(T_1,\dots,T_d) \subset \cv{R}.
\]
The other inclusion follows from $R\subset \Sigma(T_1,\dots,T_d)$ and
the fact that, since $\cQ$ is a convex set and $E$ is an affine map,
$\Sigma(T_1,\ldots,T_d)$ is convex.
\end{proof}

\begin{prop}\label{prop-joint-taketwo} If the operators
  $T_1,\dots,T_d$ pairwise commute, the set $\Sigma(T_1,\dots,T_d)$
  coincides with ${\rm Convex\,\, Hull}\,
  (\textup{JointSpec}(T_1,\dots,T_d))$.
\end{prop}
\begin{proof}
  Using the notation of Section~\ref{sec:convexity}, we consider the
  map $\alpha \mapsto \Phi_R(\alpha)$, $\alpha\in S^{d-1}$. Notice that
  \begin{equation}
    \Phi_R(\alpha) = \sup_{\norm{u}=1}\pscal{T^{(\alpha)}u}{u},
    \label{equ:phi-R}  
  \end{equation}
  where $T^{(\alpha)}:=\sum_j \alpha_j T_j$. From
  Lemma~\ref{lem:btcommuting}, we have 
\begin{equation}
    \sup_{\norm{u}=1}\pscal{T^{(\alpha)}u}{u} =
    \Phi_{\js(T_1,\dots,T_d)}(\alpha).
    \label{equ:phi-js}
  \end{equation}
 Therefore, by Lemma~\ref{lem:duality},
  \[
  \cv{R} = \cv{\js(T_1,\dots,T_d)}.
  \]
  The result now follows from Lemma~\ref{lemm:numerical-range}.
\end{proof}

\begin{remark} \label{specsing} Proposition \ref{prop-joint-taketwo}
  applied to one operator $T$ with spectrum $\sigma$ gives $\max
  \sigma(T) = \max \Sigma (T)$.
\end{remark}

The following extends Theorem \ref{theo:toeplitz}.
\begin{theorem}
  If $\{T_{j,\hbar}\}$ is any collection of semiclassical operators 
  and classical spectrum $\mathcal{S}$.  Let $J$ be a subset of
  $I$ that accumulates at $0$.  Then:
  
  \begin{itemize}
  \item[{\rm (i)}]
   From the family
  \[
  \Big\{ {\rm Convex\,\, Hull}\,\Big(\textup{JointSpec}
  (T_1,\dots,T_d)\Big) \Big\}_{\hbar \in J}
  \]
  one can recover the convex hull of $\mathcal{S}$;
  \item[{\rm (ii)}] 
 If moreover the principal symbols of $T_1,\ldots,T_d$ are bounded  then
 $$\lim_{\hbar \rightarrow 0}
 \Sigma(T_{1,\hbar},\dots,T_{d,\hbar}) \,=\, \cv{\mathcal{S}}.$$
 \end{itemize}
\end{theorem}

\begin{proof}
  Recall from~\eqref{equ:phi-R} and~\eqref{equ:phi-js} that if $\alpha
  \in S^{d-1}$,
  \begin{eqnarray}\label{eq-18-new}
    \max \sigma \Big(\sum_{j=1}^d \alpha_j T_j \Big) = 
    \max \{ \langle x,\alpha \rangle   \,\, | \,\,  x \in \Sigma(T_1,\dots,T_d)\}\;.
  \end{eqnarray}
  Now we repeat the proof of Theorem \ref{theo:general}.  Let $F$ be
  the map of principal symbols of $T_1,\dots,T_d$, and write
  $f_{\alpha}:=\langle \alpha, \, F \rangle$.  By Lemma \ref{lem:max},
  we have that $$\lim_{\hbar \to 0} \max \sigma \Big(\sum_{j=1}^d
  \alpha_j T_j \Big) = \max f_\alpha$$ which, in view of
  \eqref{eq-18-new}, reads
  \begin{equation}
    \lim_{\hbar \to 0} \Phi_{\Sigma(T_1,\dots,T_d)}(\alpha) = \Phi_{F(M)}(\alpha).
    \label{equ:pointwise-11}
  \end{equation}
  To finish the proof we repeat the proof of Theorem
  \ref{theo:bounded}, where the convex hull ${\rm Convex\,\, Hull}\,
  (\textup{JointSpec}(T_1,\dots,T_d))$ is replaced by
  $\Sigma(T_1,\dots,T_d)$, and (\ref{equ:pointwise}) by
  \eqref{equ:pointwise-11}.
\end{proof}

\vspace{2mm}

\noindent\emph{Acknowledgements.}
AP was partially
supported by NSF Grants DMS-0965738 and DMS-0635607, an NSF CAREER
Award, a Leibniz Fellowship, Spanish Ministry of Science Grants MTM
2010-21186-C02-01 and Sev-2011-0087. LP is partially supported by the National Science
Foundation grant DMS-1006610 and the Israel Science Foundation grant
509/07.   He gratefully acknowledges the warm hospitality from the University of Chicago 
where part of this project has been fulfilled. VNS is partially supported by the Institut Universitaire de
France, the Lebesgue Center (ANR Labex LEBESGUE), and the ANR NOSEVOL
grant.  He gratefully acknowledges the hospitality of the IAS.

\bibliographystyle{new} 

\addcontentsline{toc}{section}{References}
\begin{thebibliography}{300}


\bibitem{At1982} Atiyah, M.: Convexity and commuting Hamiltonians,
  {\em Bull. London Math. Soc.} {\bf 14} (1982) 1--15.

\bibitem{Bauschke-Combettes} H.H. Bauschke and P.L. Combettes:
  \textit{Convex analysis and monotone operator theory in Hilbert
    spaces,} Springer, New York, 2011.

\bibitem{Be1975} F.A. Berezin: General concept of quantization,
  \emph{Comm. Math. Phys.}, {\bf 40} 153\--174.

\bibitem{BoGu1981} L. Boutet de Monvel and V. Guillemin.  \emph{The
    Spectral Theory of Toeplitz operators}, Annals of Mathematics
  Studies 99.  Princeton University Press, Princeton, NJ, 1981.


\bibitem{BMS} M. Bordemann, E. Meinrenken, and M. Schlichenmaier:
  Toeplitz quantization of K\"{a}hler manifolds and $gl(N), N \to
  \infty$ limits, \emph{Comm. Math. Phys.} {\bf 165} (1994) 281\--296.



\bibitem{BU} D. Borthwick and A. Uribe: {Almost complex structures and
    geometric quantization}, \emph{Math. Res. Lett.} {\bf 3} (1996),
  845\--861.


\bibitem{BuBuIv2001} D. Burago, Y.  Burago, and S. Ivanov: \emph{A
    Course in Metric Geometry}, Graduate Studies in Mathematics,
  33. American Mathematical Society, Providence, RI, 2001.

\bibitem{charbonnel} A.-M. Charbonnel: {Comportement semi-classique du
    spectre conjoint d'opérateurs pseudodifférentiels qui commutent}.
  \textit{Asymptotic Anal.} \textbf{1} (1988), no. 3, 227–261.

\bibitem{Ch2003} L.Charles: Berezin-Toeplitz operators, a
  semi-classical approach, \emph{Comm. Math. Phys.} {\bf 239} (2003)
  1\--28.

\bibitem{Ch2003b} L. Charles: \newblock Quasimodes and
  {B}ohr-{S}ommerfeld conditions for the {T}oeplitz operators.
  \newblock {\em Comm. Partial Differential Equations},
  28(9-10) (2003) 1527\--1566.


\bibitem{Ch2006a} L. Charles: Toeplitz operators and Hamiltonian Torus
  Actions, {\em Journal of Functional Analysis} {\bf 236} (2006)
  299\--350.


\bibitem{ChPeVN2011} L. Charles, \'A. Pelayo, and S. V\~{u} Ng\d{o}c:
  Isospectrality for quantum toric integrable systems, arXiv:1111.5985.
  Accepted to \emph{Ann. Sci. \'Ecole Norm. Sup.}


\bibitem{CdV} Y. Colin de Verdi{\`e}re: Spectre conjoint
  d'op{\'e}rateurs pseudo-diff{\'e}rentiels qui commutent. II.  Le cas
  int{\'e}grable, \emph{Math. Z.}  {\bf 171} (1980) 51--73.

\bibitem{CdV2} Y.  Colin de Verdi{\`e}re: Spectre conjoint
  d'op{\'e}rateurs pseudo-diff{\'e}rentiels qui commutent. I.  Le cas
  non int{\'e}grable, \emph{Duke Math. J.} {\bf 46} (1979) 169--182.

\bibitem{Moser-Dacarogna} B. Dacorogna and J. Moser, \emph{On a
    partial differential equation involving the Jacobian determinant,}
  Ann. Inst. H. Poincaré Anal. Non Linéaire {\bf 7} (1990), 1–-26.

\bibitem{d}T. Delzant: Hamiltoniens p{\'e}riodiques et images convexes
  de l'application moment, {\em Bull. Soc. Math. France} {\bf 116}
  (1988) 315--339.

\bibitem{dimassi-sjostrand} M. Dimassi and J. Sjöstrand:
  \textit{Spectral asymptotics in the semi-classical limit}.  London
  Mathematical Society Lecture Note Series,
  \textbf{268}. \textit{Cambridge University Press}, Cambridge, 1999.

\bibitem{guillemin_star-product} V. Guillemin: Star products on
  compact pre-quantizable symplectic manifolds.
  \textit{Lett. Math. Phys.} \textbf{35} (1995), no. 1, 85–89.

\bibitem{GuPaUr2007} V. Guillemin, T. Paul and A. Uribe:  ``Bottom of the Well" Semi-classical 
Trace Invariants, \emph{Math. Res. Lett.} {\bf 14} (2007), 711\--719.

\bibitem{GuSt1982} V. Guillemin and S. Sternberg: Convexity properties
  of the moment mapping, \textit{Invent. Math.}, {\bf 67} (1982),
  491--513.

\bibitem{hezari-zelditch} H. Hezari and S. Zelditch: Inverse spectral
  problem for analytic $(\Z/2\Z)^n$-symmetric domains in $\R^n$.
  \textit{Geom. Funct. Anal.} \textbf{20} (2010), no. 1, 160–191.

\bibitem{hislop-sigal} P.D. Hislop and I.M. Sigal: \textit{Introduction
  to Spectral Theory},  Springer, 1995.

\bibitem{Ko1970} B. Kostant: Quantization and unitary
  representations. I. Prequantization.  \emph{Lectures in modern
    analysis and applications, III,} pp. 87\--208. Lecture Notes in
  Math., Vol. 170, Springer, Berlin, 1970.


\bibitem{IaSjZw2002} A. Iantchenko, J. Sj{\"o}strand and M. Zworski:
  Birkhoff normal forms in semi-classical inverse problems, {\em
    Math. Res. Lett.} {\bf 9} (2002) 337--362.


\bibitem{MM} Ma, X., Marinescu, G., {Toeplitz operators on symplectic
    manifolds}, \emph{J. Geom. Anal.} {\bf 18} (2008), 565\--611.

  


\bibitem{Nie} J. Niechwiej: Support of the joint spectral measure,
  \emph{J. Austral. Math. Soc. (Series A)} {\bf 43}(1987), 74--80.

\bibitem{PeVN2009} {\'A}. Pelayo and S.  V\~u Ng\d oc: Semitoric
  integrable systems on symplectic $4$\--manifolds,
  \emph{Invent. Math.} {\bf 177} (2009) 571\--597.

\bibitem{PeVN2011} {\'A}. Pelayo and S.  V\~u Ng\d oc: Constructing
  integrable systems of semitoric type, \emph{Acta Math.} {\bf 206}
  (2011) 93\--125.

\bibitem{reed-simon-1} M. Reed and B. Simon: \emph{Methods of Modern
    and Mathematical Physics, I : Functional Analysis}, Academic
  Press, 1980.

\bibitem{SaZh1999} D.A. Sadovski{\'\i} and B.I. Zhilinski{\'\i}:
  \newblock Monodromy, diabolic points, and angular momentum coupling.
  \newblock {\em Phys. Lett. A}, 256(4):235--244, 1999.

\bibitem{Kunneth} J.H. Sampson and G.A. Washnitzer: \emph{ A Künneth
    formula for coherent algebraic sheaves}, Illinois J. Math. {\bf 3}
  (1959), 389–-402.


\bibitem{ShZe2002} B. Shiffman and Z. Zelditch: Asymptotics of almost
  holomorphic sections of ample line bundles on symplectic manifolds.
  \textit{J. Reine Angew. Math.} \textbf{544} (2002), 181–222.

\bibitem{So1970} J.-M. Souriau: \emph{Structure des Syst{\`e}mes
    Dynamiques}. Maitrises de math{\'e}matiques Dunod, Paris 1970
  xxxii+414 pp.

\bibitem{VN2007} S. V\~{u} Ng\d{o}c: Moment polytopes for symplectic
  manifolds with monodromy, \textit{Adv. Math.} \textbf{208}(2)
  (2007), 909--934.

\bibitem{san-cpam} S. V\~u Ng\d oc: Bohr-Sommerfeld conditions for
  integrable systems with critical manifolds of focus-focus type.
  \textit{Comm. Pure Appl. Math.} \textbf{53} (2000), no. 2, 143–217.

\bibitem{san-inverse} S. V\~{u} Ng\d{o}c: Symplectic inverse spectral
  theory for pseudodifferential operators.  Geometric Aspects of
  Analysis and Mechanics \emph{Progress in Mathematics}
  Birkh{\"a}user/Springer, New York, 2011, Volume 292, 353\--372.

\bibitem{Z} M. Zworski, \textit{Semiclassical Analysis},
  Amer. Math. Soc., 2012.



\end{thebibliography}

\noindent
\medskip\noindent

\noindent
\noindent
\\
{\bf {\'A}lvaro Pelayo} \\
School of Mathematics\\
Institute for Advanced Study\\
Einstein Drive\\
Princeton, NJ 08540 USA.
\\
\\
\noindent
Washington University,  Mathematics Department \\
One Brookings Drive, Campus Box 1146\\
St Louis, MO 63130-4899, USA.\\
{\em E\--mail}: \texttt{apelayo@math.wustl.edu} \\
{\em Website}: \url{http://www.math.wustl.edu/~apelayo/}

\medskip\noindent

\noindent
\noindent
{\bf Leonid Polterovich} \\
School of Mathematical Sciences\\
Tel Aviv University\\
69978 Tel Aviv\\
Israel \\
{\em E\--mail}: \texttt{polterov@runbox.com}
\\
{\em Website}: \url{http://www.math.tau.ac.il/~polterov/}

\medskip\noindent

\noindent
\noindent
{\bf San V\~u Ng\d oc} \\
Institut Universitaire de France
\\
\\
Institut de Recherches Math\'ematiques de Rennes\\
Universit\'e de Rennes 1\\
Campus de Beaulieu\\
F-35042 Rennes cedex, France\\
{\em E-mail:} \texttt{san.vu-ngoc@univ-rennes1.fr}\\
{\em Website}: \url{http://blogperso.univ-rennes1.fr/san.vu-ngoc/}

\end{document}